\documentclass[runningheads]{llncs}

\usepackage{amsmath,amsfonts}
\usepackage{mathrsfs}
\usepackage{tikz-cd}
\usepackage{hyperref}
\usepackage{comment}
\usepackage{graphicx}

\usepackage[]{mdframed}

\newcommand{\Z}{{\mathbb Z}}

\newcommand{\fd}{\mathscr{F}in\mathscr{D}ist}

\newcommand{\cC}{\mathcal{C}}

\newcommand{\cE}{\mathcal{E}}

\newcommand{\cH}{\mathcal{H}}

\newcommand{\cL}{\mathcal{L}}
\newcommand{\cM}{\mathcal{M}}

\newcommand{\sH}{\mathscr{H}}

\newcommand{\sS}{\mathscr{S}}

\newcommand{\sC}{\mathscr{C}}

\newcommand{\bridge}{\mathbf{B}}
\newcommand{\SK}{\mathcal{SK}}
\newcommand{\PK}{\mathcal{PK}}

\newcommand{\negl}{\mathrm{negl}}

\DeclareMathOperator{\Enc}{\mathrm{Enc}}
\DeclareMathOperator{\Dec}{\mathrm{Dec}}
\DeclareMathOperator{\KeyGen}{\mathrm{KeyGen}}
\DeclareMathOperator{\Eval}{\mathrm{Eval}}

\begin{document}

\title{Composing Bridges}
%
%
\author{Mugurel Barcau\inst{1,2} \and
Vicen\c tiu Pa\c sol \inst{1,2} \and George C. \c Turca\c s \inst{1,3}}
\authorrunning{M. Barcau, V. Pa\c sol and G. C. \c Turca\c s}
%
\institute{certSIGN -- Research and Innovation, Bucharest, Romania 
\and
Institute of Mathematics ``Simion Stoilow" of the Romanian Academy
 \and Babe\c s-Bolyai University, Cluj-Napoca, Romania
\email{\{alexandru.barcau,vicentiu.pasol,george.turcas\}@certsign.ro}}

\maketitle              
\begin{abstract} 

The present work builds on previous investigations of the authors (and their collaborators) regarding bridges, a certain type of morphisms between encryption schemes, making a step forward in developing a (category theory) language for studying relations between encryption schemes. Here we analyse the conditions under which bridges can be performed sequentially, formalizing the notion of composability. One of our results gives a sufficient condition for a pair of bridges to be composable. We illustrate that composing two bridges, each independently satisfying a previously established IND-CPA security definition, can actually lead to an insecure bridge. Our main result gives a sufficient condition that a pair of secure composable bridges should satisfy in order for their composition to be a secure bridge. We also introduce the concept of a \textit{complete} bridge and show that it is connected to the notion of Fully composable Homomorphic Encryption (FcHE), recently considered by Micciancio. Moreover, we show that a result of Micciancio which gives a construction of FcHE schemes can be phrased in the language of complete bridges, where his insights  can be formalised in a greater generality.

\end{abstract}

\section{Introduction}

\noindent\textbf{1.1} When designing a complex cryptographic solution, one has to combine multiple cryptographic protocols and the interaction between these protocols is playing an important role in the security of the global solution. For example, in such a solution one might have ciphertexts encrypted under different secret keys, or even encrypted using different encryption schemes. This gives rise to the necessity of switching ciphertexts encrypted using one secret key to ciphertexts encrypted under a different key. A solution can be found in the literature under the name of Proxy Re-Encryption (see \cite{Do03} and the references within). A similar idea emerges in Hybrid Homomorphic Encryption (see for example \cite{Rasta18} and \cite{Pasta21}). Such a protocol is used mainly to reduce bandwidth costs resulting from ciphertext expansion for the homomorphic encryption schemes. When Hybrid Homomorphic Encryption is deployed, a server is able to convert ciphertexts encrypted using a symmetric cipher to ciphertexts encrypted using a homomorphic encryption scheme.  In our previous work \cite{BLPT23}, we defined and gave examples of \textit{bridges}, formalizing the conditions under which an algorithm that publicly transforms encrypted data from one scheme to another should perform. In the same work, we defined the IND-CPA security of a bridge by relating it to the security of a specific encryption scheme associated to the bridge. 

\medskip

\noindent\textbf{1.2} Here we take a step forward and study the context of interactions between bridges and deal with their security from a global point of view in relation with their individual security considerations. For a fair general treatment of this situation, we point out to the work of Canetti and his collaborators (\cite{Can02},\cite{Can07},\cite{Can20}).

 \noindent In this work we are concerned with the sequential evaluation, i.e. composition, of multiple bridges. Ideally, after composing two bridges, the resulting ciphertext should preserve the underlying plaintext after decryption. The notion of composability we use and define here for bridges is finer than the mere instantiation to our case of the universal composition (UC) notion of Canetti. The universality UC theorem does not apply in our setting, and neither do the modularity results that can be deduced from the aforementioned theorem. In our Example \ref{ex:halfsk}, both components (protocols) are secure, while their composition fails this requirement. Moreover, the correctness property (which is seen in the ideal-process framework as a security feature) also fails for the generic composition of bridges. 
However, we remark that our main theorems which give sufficient conditions for the composition of bridges to be correct and secure resemble the ideal-process structural shape in the work of Canetti \cite{Can02}, namely, the indistinguishability from a process that can be (publicly) described.

\medskip
\noindent\textbf{1.3} Let us recall some practical and theoretical applications of bridges. Proxy Re-Encryption can be used in protocols for secure distribution of files \cite{AFGH06}, e-mail forwarding and secure payments. As pointed out in \cite[page 2]{PWANB16}, a procedure called ``key rotation", based on Proxy Re-Encryption is required by Payment Card Industry Data Security Standard (PCI DSS) and by the Open Web Application Security Project (OWASP).

In practice, by construction, all Proxy Re-Encryption (PRE) protocols can convert ciphertexts between the same scheme. Recently, the theoretical concept of Universal Proxy Re-Encryption (UPRE) \cite{DN21} was proposed in order to extend the concept of PRE by accommodating scenarios in which a delegate can convert ciphertexts from a PKE scheme into ciphertexts of a  potentially different PKE scheme. Unfortunately, UPREs are very difficult to realize in practice.

The overwhelming majority of the examples discussed here involve Fully Homomorphic Encryption (FHE) schemes. Within this setting, \textit{bootstrapping} procedures (such as the recent bootstrapping in FHEW-like cryptosystems \cite{MP21,Lee22}) are of great importance. As we pointed out in \cite{BLPT23}, these give rise to bridges. The theoretical language of bridges developed in the present work offers a conceptual framework for explaining the following previously known fact about certain homomorphic encryption schemes: 
if one performs a  circuit on encrypted data (which is a bridge) and then a bootstrapping procedure (which is also a bridge), then the correctness of decryption is preserved (hence one obtains a third bridge).

In the previous situation, the IND-CPA security analysis of the third bridge reduces to the \textit{circular security} of the homomorphic encryption scheme.
\medskip

\noindent\textbf{1.4} The present article is a natural continuation of the work in \cite{BLPT23}, where the underlying motivation was to create a category whose objects are encryption schemes, aiming to understand relations among them. Considering bridges as viable candidates for the arrows in such a category, one is naturally lead to the notion of composing bridges. As we mention in Remark \ref{rem: categ} bellow, one can perform composition of bridges when restricting to bridges that are \textit{complete}. Unfortunately, composition is not compatible with yet another feature of encryption schemes, namely their semantic security. Thus, we are forced, for the moment, to  further restrict the morphisms in our envisaged category to \emph{Gentry-type} bridges as Theorem \ref{thm:composition} asserts. Sadly, this constraint on morphisms forces a restriction on the objects  (i.e. encryption schemes), namely in this setting one must restrict to the study of relations between FHE schemes. Moreover, if we want to allow ``self-morphisms" (there is a nuance that differentiates them from endomorphisms), i.e. bridges between the same scheme which have the \emph{same} secret key, then we need to limit our attention to those FHE schemes which are circular secure (see Theorem \ref{thm:Micciancio}).

\medskip

\noindent\textbf{1.5} As we shall see throughout this paper, multiple technical difficulties arise when one sequentially performs bridges. 

Firstly, as we point out shortly after Definition \ref{def:compbridges}, due to the possible loss of correctness of decryption, the composition of two bridges is not necessarily a bridge. We overcome this by introducing the notion of \textit{complete bridges} (see Definition \ref{def:bridge_comp}), a type of bridge which satisfies an enhanced correctness property with respect to the decryption algorithms. We then show (in Proposition \ref{PropComp}) that complete bridges can be composed.

Secondly, in Section \ref{sec: seccomp} we observe that, even if one is able to compose two \textit{secure bridges}, the resulting bridge might not be secure (see Example \ref{ex:halfsk}). Theorem \ref{thm:main} gives a sufficient condition on a pair of composable secure bridges for their composition to be secure.

\medskip

\noindent\textbf{1.6} We apply our theoretical results to certain types of bridges that arise in the context of homomorphic encryption schemes. In such schemes, some boolean circuits can be evaluated on encrypted data and such an evaluation is called homomorphic. We show that the homomorphic evaluation of a circuit gives rise to a bridge between a \emph{fiber power} of the encryption scheme to the scheme itself (see Section \ref{sec:BrigCirc}).

In a recent talk, Micciancio \cite{Mic22} points out that the definition of fully homomorphic encryption does not guarantee that one can sequentially homomorphically evaluate two circuits while preserving correct decryption.  In this paper, we formulate the previous aspect in the language of bridges. The mere definition of a correct bridge does not guarantee that one can further apply another bridge to the resulting ciphertext of the first bridge without the risk of loosing correctness. Motivated by the definitional issue raised above, Micciancio introduces the notion of fully composable homomorphic encryption (FcHE) and inspired by the bootstrapping procedure (used to construct FHE schemes) he sketches \cite{Mic22} the proof of a theorem asserting sufficient conditions for the existence of a FcHE scheme. The security of the FcHE scheme constructed by Micciancio follows from the circular security of a FHE.  To address an analogous issue, building up on previous work in \cite{BLPT23}, we introduce the definition of a \textit{complete} bridge.  Finally, we show that the aforementioned result of Micciancio \cite{Mic22} can be phrased in the language of bridges, where its proof becomes more conceptual. In our proof, we constructed a special type of bridge inspired from the \textit{Gentry type} bridges defined in \cite{BLPT23}. The difference between the former and later bridges is subtle and resides in their KeyGen algorithms. The KeyGen algorithm for the bridge needed in this proof satisfies an additional condition which accounts for the fact that the security of this bridge is equivalent to the circular security of a homomorphic encryption scheme.

\medskip

\noindent\textbf{1.7}
The paper is organized as follows. In the next section we give some preliminaries used throughout the paper. In Section \ref{sec:defex} we recall the main concepts needed in this article. We define complete and composable bridges in Section \ref{sec:CompComp} and prove that complete bridges are composable. In Section \ref{sec: seccomp} we investigate the security of composition of bridges and prove our main result. The next section is dedicated to Gentry-type bridges and compositions of general bridges with such bridges, proving the correctness and security of these protocols.
In the last section, we  explicitly construct bridges from circuits and show how one can translate the original discussion of Micciancio about composition of homomorphic evaluation of circuits into our language of bridge composition. 

\medskip

\noindent\textbf{1.8}
In his talk, Micciancio explained that one of the motivations of his analysis was the connection between circular security and fully composable homomorphic encryption. The biggest hope is to place circular security into a larger setting, which will then allow to possibly prove that circular security reduces to standard assumptions. We are informally asking if the language of bridges and bridge composition adds valuable insights into this problem. In addition, the theory of composable bridges can be related to the way the third generation of FHE schemes realise the bootstrapping procedure. Indeed, the accumulator proposed in \cite{DM15} can be viewed as a composition of bridges, one of which is a Gentry type bridge. It would be interesting to further investigate the connections between this theory and the theory of fully homomorphic encryption schemes constructed on the LWE assumption.

\subsection*{Acknowledgements}

The authors are indebted to George Gugulea, Cristian Lupa\c scu and Mihai Togan for helpful discussions and comments during the preparation of this work.

\section{Preliminaries}

In all our definitions, we denote the security parameter by $\lambda$. We say that a function $\negl: \mathbb{N} \rightarrow [0, + \infty)$
is a negligible function if for any positive integer $c$ there exists 
a positive integer $N_c$, such that $\negl (n) < \dfrac{1}{n^c}$ for all $n \geq N_c$.

Throughout this article, we will use the language of finite distributions (see section 2.1 of \cite{BLPT23}). In particular, an encryption scheme comes with the following finite distributions: secret keys, public keys and the encryptions of every fixed message. To be precise, all of the distributions mentioned here are defined for each fixed value of the security parameter $\lambda$. For example, the secret key distributions $\{ \SK_{\lambda}: \lambda \in \mathbb N \}$ form an ensemble of finite distributions. Such ensembles form a category, in which one can define finite products and study relations between its objects. To fix ideas, we choose to recall the following illustrative example of a relation between two finite distributions. Given $\lambda$, the key generation algorithm of an encryption scheme gives rise to the distribution of secret keys $\mathcal{SK}_{\lambda}$ and, for each sample $sk \leftarrow \SK_{\lambda}$, the second part of the algorithm outputs a sample $pk$ from a distribution of public keys. We say that the resulted distribution of public keys $\PK_{\lambda}$ is an $\SK_{\lambda}$-distribution. More details about the language of finite distributions and the relations between the distributions associated to an encryption scheme are explained in \cite{BLPT23}.

In an effort to keep the notation simple, we sometimes omit the subscript $\lambda$ when referring to a certain ensemble of finite distributions. We will use upper case calligraphic letters for the name of the distributions and lower case italic letters for samples from various distributions.

Every PPT algorithm gives rise to a probability distribution. Throughout this work, every written identity that involves the output of a PPT algorithm is assumed to hold with overwhelming probability over the randomness introduced by the PPT algorithm. This means that the probability of failure of that identity is negligible. For brevity, we shall avoid the repeated use of the phrase ``with overwhelming probability" to indicate this situations.

\section{Definitions}\label{sec:defex}

We first recall two technical definitions about ensembles of finite distributions discussed in more detail in \cite{BLPT23}.

\begin{definition} \label{efsampl} An ensemble $\{\mathcal{X}_{\lambda} \}_{\lambda}$ of finite distributions is polynomial-time constructible if there exists a PPT algorithm $A$ such that $A(1^{\lambda})=\mathcal{X}_{\lambda}$, for every $\lambda$. An $\{\mathcal{X}_{\lambda} \}_{\lambda}$-ensemble of finite distributions $\{(\mathcal{Y}_{\lambda}, \varphi_{\lambda}:\mathcal{Y}_{\lambda} \rightarrow \mathcal{X}_{\lambda}) \}_{\lambda}$ is polynomial-time constructible on fibers if there exist a PPT algorithm $A$, such that for any $x_{\lambda} \leftarrow \mathcal{X}_{\lambda}$ we have $A(1^{\lambda}, x_{\lambda}) = \mathcal{Y}_{\lambda}|_{\mathcal{X}_{\lambda} = x_{\lambda}}$, where $\mathcal{Y}_{\lambda}|_{\mathcal{X}_{\lambda} = x_{\lambda}}$ is the fiber distribution over $x_{\lambda}$.
\end{definition}

We will also use the following notion of computational (or polynomial) indistinguishability from \cite{GM82} and \cite{Go90}.
\begin{definition} \label{def:ind} Two ensembles of finite distributions $\{\mathcal{X}_{\lambda} \}_{\lambda}$ and $\{\mathcal{Y}_{\lambda} \}_{\lambda}$ are called computationally indistinguishable if for any PPT distinguisher $D$, the quantity $$\vert \mathrm{Pr}\left\{D(\mathcal{X}_{\lambda})=1  \right\} - \mathrm{Pr}\left\{ D(\mathcal{Y}_{\lambda})=1 \right\}\vert$$ is negligible as a function of $\lambda$.
\end{definition}

Next, we review some notions related to public key encryption schemes, homomorphic encryption and bridges between encryption schemes.

\begin{definition}[$PKE$]
A public key encryption  scheme consists of three PPT algorithms

\begin{equation*}
\mathcal{E} = ({\rm KeyGen}, {\rm Enc}, {\rm Dec})  
\end{equation*}

\noindent as follows:

\begin{itemize}
\item[$\bullet$] {\rm KeyGen:} The algorithm $(sk, pk) \leftarrow {\rm KeyGen}(1^{\lambda})$ takes a unary
representation of the security parameter and outputs a secret key $sk$ and a public key $pk$.
\item[$\bullet$] {\rm Enc:} The algorithm $c \leftarrow {\rm Enc}(pk,m)$ takes the public key $pk$ and a single message $m\in \cM$ and outputs a ciphertext $c \in \cC$.
\item[$\bullet$] {\rm Dec:} The algorithm $m^{*} \leftarrow {\rm Dec}(sk, c)$ takes the secret key $sk$ and a ciphertext $c$ and outputs a message $m^{*}$.
\end{itemize}
\end{definition}

\noindent We shall always assume that a public key encryption scheme is {\it correct}, i.e. it satisfies the following property:

\noindent {\bf Correct Decryption}: The scheme $\mathcal{E}$ is correct if for all $m \in \mathcal{M}$ and all pairs of keys $(sk, pk)$ outputted by $\text{KeyGen}(1^{\lambda})$,

\begin{equation*}
{\rm Dec}(sk, \mathrm{Enc}(pk,m)) = m,  
\end{equation*}

\noindent with overwhelming probability over the finite distribution ${\rm Enc}(pk,m)$.

\medskip

A (public key) homomorphic encryption scheme is a $PKE$ scheme such that its ${\rm KeyGen}$ algorithm outputs an additional evaluation key $evk$ (besides $sk$ and $pk$), which is used by an additional PPT evaluation algorithm ${\rm Eval}$. To be precise, there is a fourth PPT algorithm:

\begin{itemize}
    \item[$\bullet$] $\mathrm{Eval}$: The algorithm takes the evaluation key $evk$, a representation of a boolean circuit $C: \mathcal{M}^{\ell} \rightarrow \mathcal{M}$ from a set of evaluable circuits $\cL$, and a set of $\ell$ ciphertexts $c_1, ..., c_{\ell}$, and outputs a ciphertext $c^* \leftarrow {\rm Eval}(evk,C,c_1, \ldots, c_{\ell})$.
\end{itemize}

\noindent Relative to the evaluation algorithm, a public-key homomorphic encryption scheme is assumed to satisfy the following correctness property:

\noindent {\bf Correct Evaluation}: The scheme $\mathcal{E}$ correctly evaluates all boolean circuits in $\cL$ if for all keys $(sk, pk, evk)$ outputted by $\text{KeyGen}(1^{\lambda})$, for all circuits $C : \mathcal{M}^{\ell} \rightarrow \mathcal{M}$, $C \in \mathcal{L}$ , and for all $m_i \in \mathcal{M}$, $1 \leq i \leq \ell$, it holds that
\begin{equation*}
{\rm Dec}(sk, {\rm Eval}(evk, C, \Enc(pk,m_1)
, ..., \Enc(pk,m_{\ell})) = C(m_1, ..., m_{\ell}),
\end{equation*}
\noindent with overwhelming probability over the randomness of ${\rm Enc}$ and ${\rm Eval}$.

\medskip

\noindent\textbf{Remark}. Nowadays, although omitted from most definitions, it is understood that all homomorphic encryptions schemes should satisfy a certain compactness property, namely that there exists a polynomial $s = s(\lambda)$ such that the output length of ${\rm Eval}$ is at most $s$ bits long, regardless of $C \in \cL$ or the number of inputs.

\medskip

We say that a public-key homomorphic encryption scheme is a {\it fully homomorphic encryption (FHE) scheme (over $\mathcal{M}$)} if the scheme correctly evaluates all possible boolean circuits $C : \mathcal{M}^{\ell} \rightarrow \mathcal{M}$, for $\ell \in \mathbb N$.

\medskip

The authors of \cite{BLPT23} proposed a general definition for an algorithm that publicly transforms encrypted data from one scheme to another. We recall here their definition of a \textit{bridge} between two encryption schemes.

\begin{definition}\label{bridge}
Let $\mathcal{E}_i = ({\rm KeyGen}_i, {\rm Enc}_i, {\rm Dec}_i)$, $i \in \{1, 2\}$ be two PKE schemes. A bridge $\mathbf{B}_{\iota, f}$  from $\mathcal{E}_1$ to $\mathcal{E}_2$ consists of:
\begin{enumerate}
    \item A function $\iota: \mathcal{M}_1\to \mathcal{M}_2$ that is computable by a deterministic polynomial time algorithm, where $\cM_i$ is the plaintext space of the scheme $\mathcal{E}_i$.

    \item A PPT {\it bridge key generation} algorithm, which has the following three stages. First, the algorithm gets the security parameter $\lambda$ and uses it to run $\mathrm{KeyGen}_1$ in order to obtain a pair of keys $sk_1, pk_1$. In the second stage the algorithm uses $sk_1$ to find a secret key $sk_2$ of level $\lambda$ for $\mathcal{E}_2$, and then uses the second part of $\mathrm{KeyGen}_2$ to produce $pk_2$. In the final stage, the algorithm takes as input the quadruple $(sk_1, pk_1, sk_2, pk_2)$ and outputs a bridge key $bk$.
    
    \item  A PPT algorithm $f$ which takes as input the bridge key $bk$ and a ciphertext $c_1\in \mathscr{C}_1$ and outputs a ciphertext $c_2\in\mathscr{C}_2$, such that 

$$
{\rm Dec}_2(sk_2, f(bk, {\rm Enc}_1(pk_1, m)))= \iota(m).
$$

\end{enumerate}
\end{definition}

\noindent\textbf{Remark}. To simplify the discussion, when $\iota : \cM_1 \to \cM_2$ is clear from the context, we will abuse notation and write $f: \cE_1 \to \cE_2$ instead of the bridge $\bridge_{\iota,f}$.

\medskip

We now give the definition/construction of the $k$-th fiber power encryption scheme $\cE^{(k)},$ associated to any encryption scheme $\cE$ and any positive integer $k$.

\begin{definition}
\label{def: fiberpower}
Let $\mathcal{E}=({\rm KeyGen}, {\rm Enc}, {\rm Dec})$ be a PKE scheme, such that $\cM$ and $\cC$ are its plaintext and ciphertext spaces. For any positive integer $k$, the $k$-th \textit{fiber power} of $\mathcal E$ is the encryption scheme $\mathcal{E}^{(k)} = (\mathrm{KeyGen},\mathrm{Enc}^{(k)}, \mathrm{Dec}^{(k)})$, whose plaintext and ciphertext are the $k$-th Cartesian products $\cM^k$ and $\cC^{k}$ respectively. Moreover, the key generation algorithm is identical to the one of $\mathcal{E}$, outputting a pair $(sk,pk)$. The encryption and decryption algorithms of $\mathcal{E}^{(k)}$ are
$$\mathrm{Enc}^{(k)}(pk,(m_1, \ldots, m_k)) = \left(\mathrm{Enc}(pk,m_1), \ldots,  \mathrm{Enc}(pk,m_k)\right)$$
and
$$\mathrm{Dec}^{(k)}(sk,(c_1,\ldots, c_k)) = \left( \mathrm{Dec}(sk,c_1),\ldots, \mathrm{Dec}(sk,c_k) \right),$$
for any $(m_1, \ldots, m_k) \in \cM^k$ and $(c_1, \ldots, c_k) \in \cC^k$.
\end{definition}

\begin{remark}
 In the category $\fd$ of finite distributions, the encryptions of the defined scheme $\cE^{(k)}$ are products of $k$ copies of the distributions of encryptions from $\cE$ regarded as $\PK$- distributions, as defined in Section 2.1 of \cite{BLPT23}. 
\end{remark}

We end this section with a simple example of a bridge from $\cE^{(2)}$ to $\cE$, whose importance will become apparent in the next section. Moreover, construction of certain bridges from $\cE^{(k)}$ to $\cE$, for general $k$, will play a crucial role in our work presented here.  

\begin{example} We recall briefly the definition of a basic homomorphic LWE encryption scheme with plaintext space $\cM =  \Z_2$. The scheme $\mathcal{E}$ is parameterized by a dimension $n$, a ciphertext modulus $q=n^{O(1)}$ and a randomized rounding function $\chi : \mathbb R \to \mathbb Z$. In order to achieve an additive homomorphic property, we impose that the rounding function has error distribution satisfying $|\chi(x)-x| < q/8$. The secret key of the encryption scheme is a vector $sk \in \mathbb{Z}_{q}^n$, which is chosen uniformly at random. For simplicity, we will assume that this scheme is symmetric, namely that the public key $pk=sk$. The encryption of a message $m \in \Z_2$ under the key $sk \in \Z_q^n$ is given by

$$\mathrm{Enc}(pk,m) = \left(a, \chi(a \cdot sk + mq/2) \, \, \mathrm{mod} \, \,  q \right) \in \Z_q^{n+1},$$
where $a \leftarrow \mathbb Z_q^{n}$ is chosen uniformly at random. A ciphertext $(a,b)$ is decrypted as follows
$$\mathrm{Dec}(sk, (a,b)) = \lfloor 2(b-a \cdot sk)/q \rceil  \, \, \mathrm{mod} \, \, 2.$$
Note that correction of the decryption follows from the assumption on the rounding error of $\chi$. Using the homomorphic addition that can be performed on this scheme, we can easily construct a bridge from $\mathcal{E}^{(2)}$ to $\mathcal{E}$, as we now explain.

The function $\iota : \cM \times \cM \to \cM$ is defined as $\iota(m_1,m_2) = m_1 + m_2$, where the addition is performed in $\Z_2$.

The key generation algorithm of the bridge uses the key generation algorithm of $\mathcal E^{(2)}$ to output a pair $(sk,pk)$ and then, in the second stage, associates the same pair $(sk,pk)$ to $\mathcal E$. Finally, the algorithm outputs an empty bridge key $bk$. 

The ppt algorithm $f$ takes a pair of ciphertexts $(a_1, b_1), (a_2, b_2) \in \Z_q^{n+1} $ and computes
$$
f((a_1,b_1), (a_2, b_2)) := (a_1+a_2, b_1+b_2).$$
The correctness property of the bridge follows immediately from the fact that the rounding error of $\chi$ is less than $q/8$, in other words the corectness is implied by the additive homomorphic property of the scheme $\mathcal E$. 
\end{example}

\begin{remark} \label{rmk:2} Note that in the setting above, given a pair of elements in the ciphertext space $(a_1,b_1), (a_2, b_2) \in \Z_q^{n+1}$, it is not true in general that
$$\mathrm{Dec}(sk,f((a_1,b_1), (a_2, b_2))) = \mathrm{Dec}(sk,(a_1,b_1)) + \mathrm{Dec}(sk, (a_2,b_2)).$$
To see this, if $\mathrm{Dec}(sk,(a_i,b_i)) = m_i$, then the error $b_i-a_i \cdot s_k - m_iq/2 \, \, \textrm{mod} \, \, q$ might not lie in the interval $(-q/8, q/8)$, for some $i \in \{ 1,2 \}$. In this case, $b_1+b_2-(a_1+a_2) \cdot sk - (m_1+m_2)q/2 \, \, \textrm{mod} \, \, q$ might not belong to $(-q/4 , q/4)$, so that $$\mathrm{Dec}(sk,f((a_1,b_1), (a_2, b_2))) =\mathrm{Dec}(sk,(a_1+a_2, b_1+b_2))$$
is not guaranteed to be equal to $m_1+m_2$.
\end{remark}

\section{Complete and composable bridges}\label{sec:CompComp}

In general, as we noticed in the last remark of the previous section, the square of the following diagram is not commutative: 
\begin{center}
    \begin{tikzcd}
    {\rm Enc}_1(\cM_1)  \arrow[r,hook] &\cC_1 \arrow[rr, "f"] \arrow[d, "\mathrm{Dec}_1"] & & \mathcal{C}_2 \arrow[d, "\mathrm{Dec}_2"] \\
    &\mathcal{M}_1 \arrow[rr, "\iota"] & & \mathcal{M}_2 
    \end{tikzcd}
\end{center}
 where $f$ is a bridge between encryption schemes. Indeed, the definition of a bridge is imposing commutativity only for the restriction of $f$ to $\mathrm{Enc}_1(\mathcal M_1) \subseteq \mathcal C_1$, i.e. to the subset of the ciphertext space which consists of \emph{fresh encryptions}. This observation leads to the following definition. 

\begin{definition}
\label{complete:bridge}
Let $\mathcal{E}_i = (\cM_i, \cC_i, {\rm KeyGen}_i, {\rm Enc}_i, {\rm Dec}_i)$, $i \in \{1, 2\}$ be two encryption schemes. A bridge $\mathbf{B}_{\iota, f}$  from $\mathcal{E}_1$ to $\mathcal{E}_2$ is called \textbf{complete} if the PPT algorithm $f$, which takes as input the bridge key $bk$ and \emph{any} ciphertext $c_1 \in \mathcal C_1$ satisfies the equality
$$\mathrm{Dec}_2(sk_2,f(bk, c_1)) = \iota (\mathrm{Dec}_1(sk_1, c_1)),$$
for any the triple $(sk_1, sk_2, bk)$ outputted by the key generation algorithm of the bridge.
\end{definition}

\noindent Let us notice the following easy fact: 

\begin{remark}
Suppose that $\mathcal E=(\mathrm{KeyGen}, \mathrm{Enc}, \mathrm{Dec})$ is an encryption scheme. If for every pair $(sk, pk) \leftarrow \mathrm{KeyGen}(1^{\lambda})$ and every $c \in \mathcal{C}$  there exists a $m \in \mathcal M$ such that $c \in \mathrm{Enc}(pk,m)$, then any bridge from $\mathcal{E}$ to any other scheme, is complete.  
\end{remark}

As a consequence, the bridge from the Goldwaser-Micali (GM) encryption scheme to the Sander-Young-Yung (SYY) encryption scheme described in \cite{BLPT23} is complete. Indeed, the GM encryption scheme does enjoy the property described in the above remark: for any two distinct primes $p,q$, every element from the ciphertext space $J_1(pq)$ can be obtained as an encryption $\mathrm{Enc}_{GM}(pq, m)$ of some plaintext $m \in \mathbb Z_2$.

An example of a different flavor arises from the modulus switching procedure, which is frequently used in homomorphic encryption schemes for different purposes (such as reducing the noise of a ciphertext). 

\begin{example}
Let $q, Q$ be two integers such that $q$ divides $Q$. For simplicity, assume that $q , Q \equiv 2 \pmod{4}$. We let ${\rm LWE}^{2/q}$ and ${\rm LWE}^{2/Q}$ be the LWE encryption schemes with plaintext spaces $\mathcal M=\mathbb Z_2$ and ciphertext modulus $q$ and $Q$, respectively (see \cite{DM15}).

Recall that a ciphertext encrypting the message $m$ in ${\rm LWE}^{2/q}$ is of the form $(a,b)$, where $a \leftarrow \mathbb{Z}_q^n$ is sampled uniformly and $b=\langle a,s \rangle + m \cdot \frac{q}{2} + e$, where $e$ is drawn from the noise distribution.

The decryption can be regarded as a map from $\mathbb Z_q^n \times \mathbb Z_q$ to $\mathbb Z_2$, given by
$$\mathrm{Dec}_{q} \left(a,b \right) = \left\{ \begin{array}{l} 0, \text{ if } b- \langle a,s \rangle \in (-q/4,q/4) \\
1, \text{ otherwise.}\end{array} \right. $$

We define a bridge $f$ from ${\rm LWE}^{2/q}$ to ${\rm LWE}^{2/Q}$, as follows. In the second stage, the bridge key generation algorithm takes as input $s \in \mathbb Z_q^n$, the secret key of ${\rm LWE}^{2/q}$, and generates the same key for ${\rm LWE}^{2/Q}$, viewed in $\mathbb Z_Q^n$. Finally, this algorithm outputs an empty bridge key.

The algorithm $f$ will take as input a ciphertext $(a,b) \in \mathbb Z_q^n \times \mathbb Z_q$ and computes
$\left( \frac{Q}{q}  \cdot a, \frac{Q}{q} \cdot b \right) \in \mathbb Z_Q^{n} \times \mathbb Z_Q$. It is not hard to see that this bridge is complete.

This resembles the modulus-switching procedures  used initially to simplify the decryption circuit in \cite{BV11} and later to reduce the noise accumulated after homomorphic multiplication in \cite{BGV12}. We note that in modulus-switching procedures the ciphertexts are converted from a larger modulus $Q$, to a smaller one $q$, whereas in the example described above the procedures go from a smaller modulus to a larger one.

However, let us remark that, one can always modify the decryption algorithm in ${\rm LWE}^{2/q}$ to obtain a complete bridge to ${\rm LWE}^{2/Q}$ even when the modulus $Q$ is smaller than $q$.
\end{example}

It will become apparent in the following that the completeness property of a bridge plays an important role in being able to move an encryption through a chain of bridges without performing any decryption. This triggers one of the main notions presented in this paper, namely the composability of bridges which shall be discussed in details in what follows.

Given two bridges $f : \mathcal E_1 \to \mathcal E_2$, $g: \mathcal E_2 \to \mathcal E_3$, between encryption schemes $\mathcal E_1$ and $\mathcal E_2$ respectively $ \mathcal E_2$ and $\mathcal E_3$, it is desirable to be able to use them, in order to convert a ciphertext from $\mathcal E_1$ to $\mathcal E_3$ while preserving the underlying plaintext. If possible, such a procedure should resemble the composition of two functions. In particular, it is obvious that to achieve such a conversion, the Key Generation algorithms of $f$ and $g$ should not run independently, as we now explain in the following construction/definition: 

\begin{definition}[Composition of bridges] \label{def:compbridges} Let $f: \mathcal E_1 \to \mathcal E_2$ and $g: \mathcal E_2 \to \mathcal E_3$ be bridges between the encryption schemes $\mathcal E_1, \mathcal E_2$, and $ \mathcal E_2, \mathcal E_3$, respectively. The composition of $g \circ f$ consists of the following procedures:
\begin{enumerate}
    \item The function  $\iota_{g \circ f} := \iota_g \circ \iota_f : \mathcal P_1 \to \mathcal P_3$, which is computable by a deterministic polynomial time algorithm.
    \item A PPT bridge key generation algorithm, which runs as follows. First, it takes the security parameter $\lambda$ and runs the Key Generation algorithm of $f$ to obtain $(sk_1, pk_1, sk_2, pk_2, bk_f)$. After that, it only uses the second and third stages of the Key Generation of $g$ to obtain  $(sk_3, pk_3, bk_g)$. Finally, it outputs the bridge key $bk_{g \circ f} = (bk_f, pk_2, bk_g)$.
\item A PPT algorithm $g \circ f$ which takes as input a ciphertext $ c_1 \in \mathcal C_1$ and outputs $g(bk_g, f(bk_f, c_1))$.
\end{enumerate}

\end{definition}

A first observation is that, unfortunately, the composition of two bridges is not necessarily a bridge. Indeed, due to the fact that for $m \in \mathcal P_1$, the ciphertext $f(bk_1, \mathrm{Enc}_1(pk_1,m))$ is not necessarily a fresh encryption of $\iota_f(m)$, the correctness properties of $f$ and $g$ do not imply that
$$\mathrm{Dec}_{3}(sk_3, (g \circ f) (bk_{g \circ f},\mathrm{Enc}_1(pk_1,m)) = \iota_{g \circ f}(m).$$

\noindent This feature will also be present in the case of bridges constructed from circuits in section \ref{sec:BrigCirc} and it was pointed out before in a presentation by Micciancio \cite{Mic22}.

On the other hand, if the second bridge of a pair of bridges is complete, then the composition is also a bridge as the following proposition states.

\begin{proposition} \label{PropComp} Let $f: \mathcal E_1 \to \mathcal E_2$ and $g: \mathcal E_2 \to \mathcal E_3$ be bridges between the encryption schemes $\mathcal E_1, \mathcal E_2$ and $\mathcal E_3$, such that $g$ is complete. Then, the composition
$g \circ f : \mathcal E_1 \to \mathcal E_3$ is a bridge.
\end{proposition}

\begin{proof}
Since $g$ is complete, the correctness property for $g$ holds for any ciphertext. In particular, it holds for ciphertexts of the form $f(bk_1, \mathrm{Enc}_1(pk_1,m))$, so that we have:

\begin{eqnarray*}
\mathrm{Dec}_{3}(sk_3, (g \circ f) (bk_{g \circ f},\mathrm{Enc}_1(pk_1,m))) & = &\iota_g(\mathrm{Dec}_{2}(sk_2, f (bk_f,\mathrm{Enc}_1(pk_1,m)))) \\
& = & \iota_g(\iota_f(m)) = \iota_{g \circ f}(m),\\
\end{eqnarray*}
where the first equality follows from the completeness of $g$ and the second equality from the correctness property of $f$.
\end{proof}

To end the discussion of this section we introduce the following definition:

\begin{definition}\label{def:bridge_comp}
    A pair of bridges $(f,g)$ is called \emph{composable} if $g \circ f$ is defined and is also a bridge. 
\end{definition}

We have the following consequence of the last proposition:

\begin{corollary}
    If $f: \mathcal E_1 \to \mathcal E_2$ and $g: \mathcal E_2 \to \mathcal E_3$ are two complete bridges then they are composable and their composition is complete.
\end{corollary}
 
\begin{proof}
    The proof follows immediately using the same argument as in the proof of the last Proposition.
\end{proof}

\begin{remark}\label{rem: categ}
    Since complete bridges behave well with respect to composition, they form the class of morphisms of a category.
\end{remark}
\section{On the security of composition of bridges}\label{sec: seccomp}

The aim of this section is to investigate the security of the composition of two bridges. We shall introduce first the security notions related to the security of a bridge without giving all the details (for a full account see \cite{BLPT23}). Then we prove our main theorems stating that the composition of secure bridges is secure under a certain technical condition. In the next section we shall prove that this condition is satisfied by an important class of bridges.

Given a public-key encryption scheme $\mathcal{E}=(\KeyGen, \Enc, \Dec)$ and an adversary $\mathcal{A}$ consider the following the IND-CPA experiment $\mathrm{Expr}^{\mathrm{IND-CPA}}[\mathcal{A}]$:

\medskip

1. Run $\KeyGen(1^{\lambda})$ to obtain the pair of keys $(sk, pk)$.

2. The key $pk$ is given to the adversary $\mathcal{A}$. It outputs a pair of messages $m_0$, $m_1$ of its choice.

3. The challenger chooses a uniform bit $b \in \{0, 1\}$, and then a ciphertext
$c \leftarrow   \Enc(pk, m_b)$ is computed and given to $\mathcal{A}$.

4. The adversary $\mathcal{A}$ outputs a bit $b'$.

5. The output of the experiment is defined to be 1 if $b' = b$, and
0 otherwise.

\begin{definition}[IND-CPA Security] \label{IND-CPA} 
The advantage of adversary $\mathcal{A}$ against the IND-CPA security of the scheme is $\mathcal{E}$ is defined by 

\begin{equation*}
    \mathrm{Adv}^{\mathrm{IND-CPA}} [\mathcal{A}](\lambda):= | \mathrm{Pr}\{\mathrm{Expr}^{\mathrm{IND-CPA}}[\mathcal{A}] =1\} - \mathrm{Pr}\{\mathrm{Expr}^{\mathrm{IND-CPA}}[\mathcal{A}] =0 \}|,
\end{equation*}

\noindent where the probability is over the randomness of $\mathcal{A}$ and of the experiment. A public key encryption scheme $\mathcal{E}$ has indistinguishable encryptions under  chosen-plaintext attack (or is IND-CPA-secure) if for any probabilistic polynomial-time adversaries $\mathcal{A}$  there exists a negligible function $\mathrm{negl}$ such that 

\begin{equation*}
 \mathrm{Adv}^{\mathrm{IND-CPA}} [\mathcal{A}](\lambda) = \negl(\lambda).    
\end{equation*}
\end{definition}

\begin{remark}\label{rem:prob}
As it is well known (see for instance Proposition 5.9 in \cite{BR05}), an encryption scheme is IND-CPA-secure if and only if for any PPT-adversary $\mathcal{A}$ there exists a negligible function $\negl$ such that

\begin{equation*}
    \mathrm{Pr}[\mathrm{Expr}^{\mathrm{IND-CPA}}[\mathcal{A}]=1] (\lambda)\leq \frac{1}{2} + \negl(\lambda).
\end{equation*}   
\end{remark}

In \cite{BLPT23}, the authors associated to any bridge $f:\mathcal{E}_1\to \mathcal{E}_2$ an encryption scheme $\mathcal{G}_f$ (called the graph of the bridge) and defined the security of the bridge as being the security of the associated scheme $\mathcal{G}_f.$ 
On the other hand, they proved that the security of $\mathcal{G}_f$ is equivalent to the security of $\mathcal{E}_1[pk_2, bk_f]$ which is the encryption scheme $\mathcal{E}_1$ endowed with additional public information given by the public key of $\mathcal{E}_2$ and the bridge key of $f$ (for further details, see \cite[Theorem 1]{BLPT23}). 
Hereafter, we propose this equivalent notion as the security of a bridge because in practice it is much easier to work with it.

Since we will make use of the explicit construction of $\mathcal{G}_f$ in our arguments, we shall briefly recall its structure (for more details see the Section 3 of \cite{BLPT23}):
\begin{enumerate}
    \item[$\bullet$] The plaintext space of $\mathcal{G}_f$ is the same as the plaintext space of $\mathcal{E}_1$, that is $\mathcal{M}_1$. The ciphertext space is $\mathcal{C}_1\times \mathcal{C}_2,$ where $\mathcal{C}_i$ is the ciphertext space of $\mathcal{E}_i$, $\forall i \in \{ 1, 2\}.$
    \item[$\bullet$] The secret key of $\mathcal{G}_f$ is the pair: $sk_{\mathcal{G}_f}:=(sk_1,sk_2).$
    \item[$\bullet$] The public key is the triple: $pk_{\mathcal{G}_f}:=(pk_1,pk_2,bk_f).$
     \item[$\bullet$] The encryption algorithm is given by
\begin{equation*}
     \Enc_{\mathcal{G}_f}(pk_{\mathcal{G}_f}, m):=(\Enc_1(pk_1, m), f(bk_f,\Enc_1(pk_1, m)),
\end{equation*}     
\noindent for all $m \in \mathcal{M}_1.$
     \item[$\bullet$] The decryption algorithm is $\Dec_{\mathcal{G}_f}(sk_{\mathcal{G}_f}, (a,b)):=\Dec_1(sk_1, a).$
    
\end{enumerate}

In what follows we are interested in finding conditions that ensure the security of a bridge obtained as a composition of two composable bridges. 
The following easy result, can be seen as a necessary condition for the security of such a bridge.

\begin{proposition} Suppose $f: \mathcal E_1 \to \mathcal E_2$ and $g: \mathcal E_2 \to \mathcal E_3$ are composable bridges, such that $g \circ f$ is a secure bridge. Then $f$ is a secure bridge.
\end{proposition}

\begin{proof}
Since the bridge key  of $g\circ f$ is given by $bk_{g \circ f} = (bk_f, pk_2, bk_g)$, the security of $f$ follows from the above considerations.  
\end{proof}
On the other hand, if the composition $g \circ f$ is a secure bridge, it does not follow that $g$ is secure. Let us illustrate this with the following example.

\begin{example}
Assume that $f : \mathcal E_1 \to \mathcal E_2$ and $g : \mathcal E_2 \to \mathcal E_3$ are composable secure bridges such that $g \circ f$ is secure. By \cite[Proposition 1]{BLPT23}, the encryption schemes $\mathcal E_1$ and $\mathcal E_2$ are IND-CPA secure. For simplicity, assume that the common plaintext is $\mathcal M = \{0,1 \}.$ We now construct the encryption scheme $\mathcal E_2'$, as follows: this scheme is almost identical to $\mathcal E_2$, except that the encryption reveals the message. More precisely, $\mathrm{Enc}_{2}'(pk_2, m) = (\mathrm{Enc}_2(pk_2, m), m) $ is a concatenation of the encryption of $m$ in the scheme $\mathcal E_2$ with the plaintext $m$ itself. This new encryption scheme $\mathcal E_2'$ is obviously not secure.

On the other hand, we modify $f$ and $g$ to new bridges $f'$ and $g'$ in the following way. For any ciphertext $c \in \mathcal C_1$, we let $f'(c_1) = (f(c_1), b)$, where $b \leftarrow \{0,1 \}$ is chosen uniformly at random. The bridge $g'$ will perform as follows. For any ciphertext $c_2'=(c_2,b) \in \mathcal C_2'$, we let $g'(c_2') = g(c_2) \in \mathcal C_3$. It can be easily verified that $g' \circ f'$ is a secure bridge from $\mathcal E_1 \to \mathcal E_3$. However, the bridge $g'$ is not secure because the modified encryption scheme $\mathcal E_2'$ is not secure.

\end{example}

Next, we give an example of two (complete) secure bridges $f$ and $g$, which are composable, for which $g \circ f$ is an insecure bridge.

\begin{example}\label{ex:halfsk}
    Consider $\mathcal E$ an encryption scheme (not homomorphic or with any other special property). We shall assume that knowing half of the secret key does not harm the security. Denote by $\ell(\lambda)$ the bit length of its secret key. 
    We define $\mathcal E_1$ to be $\mathcal E$, while $\mathcal E_2$ is $\mathcal E$ with a modified cyphertext space, consisting of concatenations of  cyphertexts of $\mathcal E$ with strings of length $\ell(\lambda)/2$. The encryption algorithm in $\mathcal{E}_2$ outputs a concatenation of an encryption of $\mathcal{E}$ with a random sequence of length $\ell(\lambda)/2$. The decryption algorithm is just the decryption algorithm of $\mathcal{E}$ applied to the  first part of the ciphertext, i.e. the bit string corresponding to a ciphertext in the scheme $\mathcal{E}$. Also, consider $\mathcal E_3$ to be the scheme obtained by applying on $\mathcal{E}_2$ the same procedure  used in the 
    construction of $\mathcal E_2$ from $\mathcal{E}$.
    
The bridge key of $f: \mathcal E_1\to \mathcal E_2$ consists of the first half of the bit representation of the secret key of $\mathcal{E}$ and $f$ is defined by the concatenation of the identity map with the bridge key. The bridge $g:\mathcal{E}_2\to \mathcal E_3$ has a similar construction, only  that the bridge key of $g$ consists of the second half of the secret key of $\mathcal{E}$. 
    Now, by our assumptions, the schemes $\mathcal{E}[bk_f]$ and $\mathcal{E}[bk_g]$ are secure. Hence, the bridges $f$ and $g$ are secure. Moreover, $f$ and $g$ are clearly complete, thus composable.
    
     However, their composition reveals the entire secret key because the bridge key of the composition contains both halves of the secret key of $\mathcal{E}$.
 Schemes $\mathcal{E}$ satisfying the above conditions are plenty in the literature. For example any LWE scheme is such.
\end{example}

From these examples we learn that an additional condition on a pair of composable bridges is required in order to get a secure composition.

This motivates our main result, that is Theorem \ref{thm:main}. This result is heavily influenced by the main result of \cite{BLPT23}, which we recall first together with a more conceptual proof.

Suppose ${\bf B}_{\iota, f}$ is a bridge from $\mathcal{E}_1$ to $\mathcal{E}_2$. Recall that the \textit{bridge key generation algorithm} produces the following ensembles of $\{\SK_{1,\lambda} \}_{\lambda}$ distributions: $\{\PK_{1,\lambda} \}_{\lambda}$, $\{\PK_{2, \lambda} \}_{\lambda}$ and $\{\mathcal{BK}_{\lambda}\}_{\lambda}$. Let $\mathcal{F}$ be the ensemble of finite distributions of triples $(pk_1, pk_2, bk)$. Recall that $\pi_1 : \mathcal{F} \to \PK_1$ is a morphism of finite distributions, so $\mathcal{F}$ is an ensemble of $\PK_1$-distributions. In this theorem, we will make use for the first time of the Definitions \ref{efsampl} and \ref{def:ind}.

\begin{theorem} \label{thm:1}
Let ${\bf B}_{\iota, f}$ be a bridge between $\mathcal E_1$ and $\mathcal E_2$ and assume that the scheme $\mathcal{E}_1$ is IND-CPA secure. If there exists a  polynomial time constructible on fibers ensemble of ${\PK}_1$-distributions $\widetilde{\mathcal{F}}$ which is computational indistinguishable from $\mathcal{F}$, then the bridge ${\bf B}_{\iota, f}$ is IND-CPA secure.
\end{theorem}

\begin{proof}
To show that the bridge ${\bf B}_{\iota, f}$ is IND-CPA secure it is enough to show that for any attacker $\mathcal{A}$ on $\mathcal{E}_1[\PK_{\mathcal{G}_f}]$ the probability of answering correctly is bounded above by $\frac{1}{2} + \mathrm{negl}(\lambda)$ (see Remark \ref{rem:prob}). 
This probability is given by

$$\mathrm{Pr}[\mathrm{Expr}[\mathcal{A}]=1] = \frac{1}{2} \mathrm{Pr}[\mathcal{A}(\mathcal{F}, \mathrm{Enc}_1(pk_1, 0))=0] + \frac{1}{2} \mathrm{Pr}[\mathcal{A}(\mathcal{F}, \mathrm{Enc}_1(pk_1, 1))=1].$$

\noindent If $\mathcal{A}$ is any attacker on $\mathcal{E}_1[\PK_{\mathcal{G}_f}]$, then we construct an attacker $\mathcal{A}_1$ on $\mathcal E_1$ as follows: $\mathcal{A}_1$ receives the pair $(pk_1, c \leftarrow \mathrm{Enc}_1(pk_1, b))$ from the challenger and uses the sampling algorithm of $\widetilde{\mathcal{F}}$ to produce a triple $(pk_1, \alpha, \beta)$. The attacker gives $(pk_1, \alpha, \beta, c)$ to $\mathcal{A}$ and then outputs the bit received from it. Since $\mathcal E_1$ is IND-CPA secure we have

\begin{equation*}
    \mathrm{Pr}[\mathrm{Expr}[\mathcal{A}_1]=1] \leq \frac{1}{2} + \mathrm{negl}_1(\lambda),
\end{equation*}

\noindent for some negligible function $\mathrm{negl}_1(\lambda)$. We get that

\begin{align} 
\begin{split}
\label{eq1}
    \mathrm{Pr}[\mathrm{Expr}[\mathcal{A}_1]=1] &= \frac{1}{2} \mathrm{Pr}[\mathcal{A}(\widetilde{\mathcal{F}}, \mathrm{Enc}_1(pk_1, 0))=0] + \frac{1}{2} \mathrm{Pr}[\mathcal{A}(\widetilde{\mathcal{F}}, \mathrm{Enc}_1(pk_1, 1))=1] \\ &\leq \frac{1}{2} + \mathrm{negl}_1(\lambda).
\end{split}
\end{align}

We construct now an IND-CPA distinguisher $\mathcal{D}_0$ between $\mathcal{F}$ and $\widetilde{\mathcal{F}}$ as follows. The distinguisher receives a triple $(pk_1, x, y)$ from the challenger and uses $pk_1$ to compute $c \leftarrow \mathrm{Enc}_1(pk_1, 0)$, gives $(pk_1, x, y, c)$ to $\mathcal{A}$ and outputs the bit received from $\mathcal{A}$. We understand that $\mathrm{Expr}[\mathcal{D}_0]$ identifies $\mathcal{F}$ if it outputs $0$ and $\widetilde{\mathcal{F}}$, if it outputs $1$. Then we have:

\begin{align}
\label{eq2}
\begin{split}
    \mathrm{Pr}[\mathrm{Expr}[\mathcal{D}_0]=1] &= \frac{1}{2} \mathrm{Pr}[\mathcal{A}(\mathcal{F}, \mathrm{Enc}_1(pk_1, 0))=0] + \frac{1}{2} \mathrm{Pr}[\mathcal{A}(\widetilde{\mathcal{F}}, \mathrm{Enc}_1(pk_1, 0))=1] \\ &\leq \frac{1}{2} + \mathrm{negl}_2(\lambda).
\end{split}
\end{align}

Similarly, one constructs $\mathcal{D}_1$, a distinguisher between $\widetilde{\mathcal{F}}$ and $\mathcal{F}$, similarly as above, but now $c\leftarrow \mathrm{Enc}_1(pk_1, 1)$. This time, $\mathrm{Expr}[\mathcal{D}_1]$ outputs $0$ if the distinguisher identifies $\widetilde{\mathcal{F}}$ and outputs $1$ if the distinguisher identifies $\mathcal{F}$. Thus, we get:

\begin{align}
\label{eq3}
\begin{split}
    \mathrm{Pr}[\mathrm{Expr}[\mathcal{D}_1]=1] &= \frac{1}{2} \mathrm{Pr}[\mathcal{A}(\mathcal{F}, \mathrm{Enc}_1(pk_1, 1))=1] + \frac{1}{2} \mathrm{Pr}[\mathcal{A}(\widetilde{\mathcal{F}}, \mathrm{Enc}_1(pk_1, 1))=0] \\ &\leq \frac{1}{2} + \mathrm{negl}_3(\lambda).
\end{split}
\end{align}

Adding the inequalities \eqref{eq1}, \eqref{eq2}, \eqref{eq3} and using the equalities:

$$
\mathrm{Pr}[\mathcal{A}(\widetilde{\mathcal{F}}, \mathrm{Enc}_1(pk_1, 0))=0] + \mathrm{Pr}[\mathcal{A}(\widetilde{\mathcal{F}}, \mathrm{Enc}_1(pk_1, 0))=1] = 1
$$

$$
\mathrm{Pr}[\mathcal{A}(\widetilde{\mathcal{F}}, \mathrm{Enc}_1(pk_1, 1))=0] + \mathrm{Pr}[\mathcal{A}(\widetilde{\mathcal{F}}, \mathrm{Enc}_1(pk_1, 1))=1] = 1
$$
we obtain
$$ \frac{1}{2} \mathrm{Pr}[\mathcal{A}(\mathcal{F}, \mathrm{Enc}_1(pk_1, 0))=0] + \frac{1}{2} \mathrm{Pr}[\mathcal{A}(\mathcal{F}, \mathrm{Enc}_1(pk_1, 1))=1]   \leq \frac{1}{2} + \mathrm{negl}(\lambda),$$
which is exactly what we wanted to prove.
\end{proof}

Now, we are able to state and prove our main result.

\begin{theorem}\label{thm:main}
Suppose $f: \mathcal{E}_1 \rightarrow \mathcal{E}_2$, $g: \mathcal{E}_2 \rightarrow \mathcal{E}_3$ are composable bridges such that $f$ is IND-CPA secure and there exists a polynomial time constructible on fibers ensemble of $\PK_f-$distributions which is computational indistinguishable from $\PK_{g \circ f}$, then $g\circ f$ is IND-CPA secure.
\end{theorem}

\begin{proof}
Notice that the composable pair of bridges $(f, g)$ gives rise to a bridge $\widehat{g}: \mathcal{G}_f \rightarrow \mathcal{E}_3$, by the formula:

\begin{equation*}
    \widehat{g}(bk_{\hat{g}}, (a, b)) := g(bk_g, b),
\end{equation*}

\noindent where $bk_{\hat{g}} = bk_g$. Since $f$ is secure, it follows that the scheme $\mathcal{G}_f$ is secure. By our assumptions, an immediate application of Theorem \ref{thm:1} shows that $\widehat{g}$ is IND-CPA secure, which means that the scheme $\mathcal{G}_{\widehat{g}}$ is IND-CPA secure. Notice that an encryption of a message $m$ in $\mathcal{G}_{\widehat{g}}$ is in fact a triple of the form $(a, f(bk_f, b), g(bk_g,f(bk_f, c)))$, where $a, b, c$ are encryptions of $m$ in $\mathcal{E}_1$, so that any adversary $\mathcal{A}$ on $\mathcal{G}_{g \circ f}$ gives rise to an adversary $\mathcal{A}'$ on $\mathcal{G}_{\widehat{g}}$. Indeed, notice first that these encryption schemes have the same public key. On the other hand if $(a, f(bk_f, b), g(bk_g,f(bk_f, c)))$ is the triple received by $\mathcal{A}'$ from the challenger (together with the public key), then $\mathcal{A}'$ gives the public key and the pair $(a, g(bk_g,f(bk_f, c)))$ to $\mathcal{A}$ and outputs the bit received from it. It is easy to see that, since

\begin{equation*}
\mathrm{Pr}[\mathrm{Expr}^{\mathrm{IND-CPA}}[\mathcal{A}'](\lambda)=1] = \mathrm{Pr}[\mathrm{Expr}^{\mathrm{IND-CPA}}[\mathcal{A}](\lambda)=1],
\end{equation*}

\noindent if the attacker $\mathcal{A}$ breaks the IND-CPA security of $\mathcal{G}_{g \circ f}$, then $\mathcal{A}'$ also breaks the security of $\mathcal{G}_{\widehat{g}}$, and this is a contradiction.
\end{proof}

\section{Gentry type bridges} \label{sec:Gentry}

We recall the construction of Gentry type bridges from \cite{BLPT23}. Briefly, the \textit{Recrypt} algorithm, used in the bootsrapping procedure that transforms a somewhat homomorphic encryption scheme into a fully homomorphic encryption scheme (see \cite{Ge10}), can be adapted to our situation in order to give a general recipe for the construction of a bridge.

Let us consider an encryption scheme 

$$\mathcal{E} =(\KeyGen_{\mathcal E}, \Enc_{\mathcal E}, \Dec_{\mathcal E})$$ 

\noindent and a homomorphic encryption scheme

$$\mathcal{H} =(\KeyGen_{\mathcal H}, \Enc_{\mathcal H}, \Dec_{\mathcal H}, \Eval_{\mathcal H}).$$

\noindent We shall also denote by $\mathcal{M}_{\mathcal E}, \mathcal{C}_{\mathcal E}$ and $\mathcal{M}_{\mathcal H}, \mathcal{C}_{\mathcal H}$, the plaintext and ciphertext spaces of $\mathcal E$ and $\mathcal H$, respectively. 

Fix once and for all, a function $\iota : \mathcal M_{\mathcal E} \to \mathcal M_{\mathcal H}$ that is computable by a deterministic polynomial time algorithm. 

At a high level, the bridge key $bk_f$ of a Gentry-type bridge $f: \mathcal E \to \mathcal H$ consists of the bit representation of the secret key of $\mathcal E$ encrypted under the public key of $\mathcal H$. Moreover, the bridge algorithm is the homomorphic evaluation of the circuit $\iota \circ \Dec_{\mathcal E}: \mathcal C_{\mathcal E}\to \mathcal{M}_{\mathcal H}$. 

We now describe with details the algorithms of the bridge. Let us start with the key generation algorithm.

\noindent $\bullet$ In the first stage, the key generation algorithm of the bridge runs $\KeyGen_{\mathcal E}(1^{\lambda})$ to obtain a pair of keys $(sk_{\mathcal E},pk_{\mathcal E})$. The second stage of the algorithm runs independently of the first one, and it just makes use of $\KeyGen_{\mathcal H}(1^\lambda)$ to obtain $(sk_{\mathcal H}, pk_{\mathcal H})$. 

The final stage of the algorithm takes as input $(sk_{\mathcal E}, pk_{\mathcal E}, sk_{\mathcal H}, pk_{\mathcal H})$ constructed as above, and creates $bk_f$ as the vector of encryptions of the bit representation of $sk_{\mathcal E}$ under $pk_{\mathcal H}$ (the details of this are below). 

\noindent $\bullet$ The PPT algorithm $f$ mentioned in the third part of Definition \ref{bridge} is the homomorphic evaluation of the algorithm $\iota \circ \Dec_{\mathcal E}$. To rigorously  perform it, we need to realise $\iota\circ\Dec_{\mathcal E}$ as a map $\mathcal{M}_{\mathcal{H}}^{\ell} \rightarrow \mathcal{M}_{\mathcal{H}}$, and for this we use the ring structure on $\cM_{\mathcal H}$ (we shall assume that such a ring structure exists; one can avoid this assumption with little work, but since all the known FHE schemes have this property we chose to work with it). 
Suppose that the ciphertext space $\sC_{\mathcal E}$ has a representation as a subset of $\{ 0, 1\}^n$ and that the set of secret keys is a subset of $\{ 0, 1\}^e$, so that $\iota \circ \Dec_{\mathcal E}: \{ 0, 1\}^e \times \{ 0, 1\}^n \rightarrow \mathcal{M}_{\mathcal H}$. Now, we construct the map $\widetilde{\iota\circ\Dec}_{\mathcal E}: \cM_{\mathcal{H}}^e \times \cM_{\mathcal{H}}^n \rightarrow \cM_{\mathcal{H}}$ as follows. Viewing $\cM_{\mathcal{H}}$ as a subset of $\{0, 1\}^m$, we have that $\iota \circ \Dec_{\mathcal E}: \{ 0, 1\}^e \times \{ 0, 1\}^n \rightarrow \cM_{\mathcal{H}}$ is a vector $(g_1,...,g_m)$ of boolean circuits expressed using ${\rm XOR}$ and ${\rm AND}$ gates. Let $\tilde{g_i}: \cM_{\mathcal{H}}^e \times \cM_{\mathcal{H}}^n \rightarrow \cM_{\mathcal{H}}$ be the circuit obtained by replacing each ${\rm XOR}(x,y)$ gate by $x\oplus y:=2(x+y) - (x+y)^2$ and each ${\rm AND}(x,y)$ gate by $x\otimes y := x \cdot y$, where $+$ and $\cdot$ are the addition and multiplication in $\mathcal{M}_{\mathcal{H}}$. Notice that the subset of $\mathcal{M}_{\mathcal{H}}$ consisting of its zero element $0_{\mathcal H}$ and its unit $1_{\mathcal H}$ together with $\oplus$ and $\otimes$ is a realisation of the field with two elements inside $\mathcal{M}_{\mathcal H}$. In other words, if $c=(c[1],...,c[n]) \in \sC_{\mathcal E}$ and $sk_{\mathcal E}=(sk[1],...,sk[e])$ is the secret key, then $\tilde{g_i}(sk[1]_{\mathcal H},...,sk[e]_{\mathcal H}, c[1]_{\mathcal H},...,c[n]_{\mathcal H}) = m_{\mathcal H}$ if $g_i(sk[1],...,sk[e],c[1],...,c[n]) = m$ for all $i$, where $m \in \{0, 1\}$.
For an element $x \in \cM_{\mathcal H}$, we let $[x=1_{\mathcal H}]$ be the equality test, which returns $1$ if $x=1_{\mathcal H}$ and $0$ otherwise. 
Finally, $\widetilde{\iota \circ \Dec}_{\mathcal E}:\mathcal{M}_{\mathcal{H}}^e \times \mathcal{M}_{\mathcal{H}}^n \to \{0,1\}^m$ is defined by: $$\left(\left[ \tilde{g_i}(y_1,..., y_e, x_1,..., x_n) = {1_\mathcal H} \right]\right)_{i = \overline{1, m}}.$$ 
One can verify immediately that 
$$\widetilde{\iota \circ \Dec}_{\mathcal E}(sk[1]_{\mathcal H},...,sk[e]_{\mathcal H}, c[1]_{\mathcal H},...,c[n]_{\mathcal H})= \iota\circ {\rm Dec}_{\mathcal E}(sk_{\mathcal E},c).$$

\noindent Now we are ready to define the bridge algorithm $f$. Given a ciphertext $c \in \mathscr{C}_{\mathcal E}$, the algorithm $f$ first encrypts the bits of $c$ (viewed as elements of $\cM_{\mathcal H}$) under $pk_{\mathcal H}$ and retains these encryptions in a vector $\tilde{c}$. The bridge key $bk_f$ is obtained by encrypting the bits $sk[i]_{\mathcal E}$ under $pk_{\mathcal H}$, for $i \in \overline{1, e}$. Then, the algorithm $f$ outputs $\mathrm{Eval}_{\mathcal H}(evk_{\mathcal H},\widetilde{\iota \circ \Dec}_{\mathcal E}, (bk_f, \tilde c)).$

\begin{remark}
Let us notice that in the above construction, one does not necessarily need to encrypt the bit representation of the ciphertext $c$ under $pk_{\mathcal H}$. The whole construction works if one homomorphically evaluates the circuits $\widetilde{\iota \circ \Dec}_{\mathcal E}(\cdot, c):\mathcal{M}_{\mathcal{H}}^e \to \mathcal{M}_{\mathcal{H}}$, for each fixed ciphertext $c$. While this variant is often more efficient, one needs to compute the circuit   $\widetilde{\iota \circ \Dec}_{\mathcal E}(\cdot, c)$ every time the bridge is applied. Also, most of the time, there exist "trivial" encryptions of $0_{\mathcal H}$ and $1_{\mathcal H}$, so that the two bridges become identical.
\end{remark}

\begin{theorem}\label{thm:Gentry_complete}
 Any Gentry type bridge is complete.   
\end{theorem}

\begin{proof}
With the above notations, if $c$ is \emph{any} ciphertext in $\mathcal{C}_{\mathcal E}$, then $\tilde{c}$ consists of an $n$-dimensional vector of "fresh" encryptions of the scheme $\mathcal{H}$. Since $bk_f$ is an $e$-dimensional vector consisting of fresh encryptions, we have:

\begin{eqnarray*}
 \Dec_{\mathcal H} (sk_{\mathcal H}, f(bk_f,c)) & = & \Dec_{\mathcal H}  \left( sk_{\mathcal H}, \Eval_{\mathcal H}(evk_{\mathcal H},\widetilde{\iota \circ \Dec}_{\mathcal E}, (bk_f, \tilde c)) \right)  \\
& = &  (\iota \circ \Dec_{\mathcal E}) \left(\Dec_{\mathcal H}(sk_{\mathcal H}, bk_f), \Dec_{\mathcal H}(sk_{\mathcal H}, \tilde{c}) \right)\\
& = & \iota \left( {\rm Dec}_{\mathcal E} (sk_{\mathcal E}, c) \right)
\end{eqnarray*}
which shows that third condition in the definition of a bridge is satisfied for any ciphertext $c \in \mathscr{C}_{\mathcal E}$, i.e. $f$ is complete.
\end{proof}

An immediate consequence of this theorem and Proposition \ref{PropComp} is the following:

\begin{corollary} \label{cor:composable}
If $f: \mathcal{E}_1 \rightarrow  \mathcal{E}_2$ is a bridge and $g: \mathcal{E}_2 \rightarrow \mathcal{H}$ is a Gentry type bridge, then the two bridges are composable. 
\end{corollary}

Now, we can state and prove the main result of this section.

\begin{theorem} \label{thm:composition} 
If $f: \mathcal{E}_1 \rightarrow  \mathcal{E}_2$ is a secure bridge, $g: \mathcal{E}_2 \rightarrow \mathcal{H}$  is a Gentry type secure bridge, and $\mathcal{H}$ is a secure FHE scheme, then $g \circ f$ is a secure bridge.
\end{theorem}

\begin{proof}
The security of the bridge will follow from Theorem \ref{thm:main}, after showing that the ensemble of distributions $\PK_{g \circ f}$ satisfies the required indistinguishably condition with respect to a certain polynomial-time constructible on fibers $\PK_f$-distribution.
Let us recall that $\PK_f$ consists of the tuple $(pk_1,pk_2, bk_f)$ and that  $\PK_{g \circ f}$ consists of $(pk_1,pk_2,bk_f, pk_{\mathcal H}, bk_{g})$. As $g$ is of Gentry-type, $bk_g$ consists of a vector encrypting the bit representation of $sk_{2}$ under $pk_{\mathcal H}$. Let $\tilde{\mathcal F}$ be the ensemble of distributions $(pk_1, pk_2, bk_f, pk_H, \tilde{bk})$, obtained in the following way. Firstly, one samples $(pk_1,pk_2, bk_f)$ from the distribution $\PK_f$. Secondly, one uses $\mathrm{KeyGen}_{\mathcal H}$ to sample 
$pk_{\mathcal{H}}$ and then we let $\widetilde{bk}=(\widetilde{bk}[1],...,\widetilde{bk}[e])$ with $\widetilde{bk}[i] \leftarrow {\rm Enc}(pk_\sH, 0_\sH)$ for all $i\in \overline{1,e}$, where $e$ is the bit-length of $sk_2$. Note that $e$ can be considered public knowledge, as it only depends on the parameters of encryption in the scheme $\mathcal{E}_2$. We also note that since $g$ is a Gentry type bridge, the public key $pk_{\mathcal H}$ and any sample from the distribution $\PK_f$ are chosen independently from each other. Note that $\tilde{\mathcal F}$ is a polynomial-time constructible on fibers $\PK_{f}$-distribution. If the scheme $\mathcal H$ is IND-CPA secure, then one can prove by a hybrid argument, identical to the one in the proof of Proposition 2 of \cite{BLPT23}, that the $\PK_f$-distributions $\PK_{g \circ f}$ and $\tilde{\mathcal F}$ are computationally indistinguishable. 
\end{proof}

\begin{remark} \label{secure:Gentry}
Notice that the Gentry type bridge $g$ in the above theorem is already secure if the schemes $\mathcal{E}_2$ and $\mathcal{H}$ are secure (cf. Theorem 3 in \cite{BLPT23}).
\end{remark}

\section{Bridges from circuits and Micciancio's Theorem}\label{sec:BrigCirc}

In this section we explain how a bridge can be canonically associated to a pair consisting of an encryption scheme and a circuit that can be homomorphically evaluated. As a consequence, one can translate every discussion about homomorphic circuit evaluation and composition of circuits into the language of bridges developed here and in \cite{BLPT23}. Moreover, one can see Theorem \ref{thm:Gentry_complete} as a generalisation of the following theorem:

\begin{theorem}[Micciancio]\label{thm:Micciancio} Every circular secure FHE encryption scheme $\mathcal{H} = (\mathrm{KeyGen}, \Enc, \Dec, \Eval)$, can be transformed into a secure fully composable homomorphic encryption scheme $\mathcal{H}' = (\mathrm{KeyGen}', \Enc, \Dec, \Eval')$. 
\end{theorem}

The terminology \textit{fully composable} is clarified in Definition \ref{fche} below.  A sketch of the proof of this theorem can be found in  the talk \cite{Mic22} \footnote{This is a talk given by D. Micciancio at the FHE.org conference 2022, that took place in Trondheim, NO - an affiliated event to the Eurocrypt 2022 conference. The proof starts at time 25:39.}. The idea of the proof is to use a procedure similar to the one in Gentry-type bridge construction. We saw in Theorem \ref{thm:Gentry_complete} that these give rise to complete bridges. The difference between the Gentry-type bridge construction and Micciancio's resides in the $\KeyGen$ algorithms. In the key generation process of a Gentry type bridge, it is required that the keys of the second scheme are generated independently from the keys of the first one. On the other hand, in Micciancio's construction the two schemes have identical secret and public keys. This difference impacts only the security of the bridge and does not affect its completeness property. This can give rise to security issues for which one is forced to add an extra assumption, namely the circular security assumption.

We start by showing how one can associate bridges to circuits and homomorphic encryption schemes.

Let $\mathcal{E} = (\KeyGen, \Enc, \Dec)$ be an encryption scheme and let $C:\mathcal{M}^r \to \mathcal{M}$ be a boolean circuit defined over the plaintext $\mathcal{M}$ of $\mathcal{E}.$ 

Assume that the scheme $\mathcal{E}$ is $C-$homomorphic. Recall that a scheme is homomorphic with respect to the circuit $C$, or is $C-$homomorphic, if there exists an algorithm $\Eval(evk, C, c_1,\ldots c_r)$ which takes as input an evaluation key $evk$, the circuit $C$, and an element of $\mathcal{C}^{r}$ such that the outputted ciphertext satisfies the following correctness condition:
$$
\Dec(sk,\Eval(evk, C, \Enc(m_1),\ldots, \Enc(m_r))) = C(m_1,\ldots,m_r),
$$ 

\noindent for all $(m_1,\ldots, m_r)\in \mathcal{M}^r.$

In this context, we associated to the pair $(\cE, C)$, the bridge $\bridge_{C} :\mathcal{E}^{(r)}\to \mathcal{E}$ where the function $\iota$ is given by:
$$
\iota(m_1,\ldots,m_r):=C(m_1,\ldots,m_r),
$$
the bridge key $bk$ consists of  the evaluation key $evk$ (possibly empty), and the bridge algorithm is
    
$$
    f_C(bk,c_1, \ldots, c_r):=\Eval(evk, C, c_1, \ldots, c_r).
$$

One can easily see that the correctness property of $\Eval$ guarantees that the above construction is indeed a bridge. Moreover, this bridge is secure as long as the $C$-homomorphic scheme $\mathcal{E}$ is secure. 

\begin{remark} If $\mathcal{E}$ is an FHE scheme, the above procedure gives rise to a family of bridges $\bridge_C$, one for each boolean circuit $C$. By concatenating such bridges $\bridge_{C_1}, \ldots, \bridge_{C_s} : \mathcal E^{(r)} \to \mathcal E$, one obtains a bridge $\bridge_{C_1,\ldots,C_s}:\mathcal{E}^{(r)}\to \mathcal{E}^{(s)}$.
\end{remark}

In his talk \cite{Mic22}, Micciancio points out one foundational problem in the definition of a fully homomorphic encryption scheme. Namely, the definition of an FHE scheme does not guarantee correct decryption if one sequentially evaluates two (multivalued) circuits on encrypted data. This issue is an instance of the more general problem of composing bridges.  

More precisely, suppose bridges $\bridge_{C_1, \dots, C_s} : \mathcal{E}^{(r)} \to \mathcal{E}^{(s)}$ and $\bridge_{D_1, \ldots, D_t} : \mathcal{E}^{(s)} \to \mathcal{E}^{(t)}$ are constructed as above, for some FHE scheme $\mathcal E$. The issue raised by Micciancio is equivalent to the composability of the pair of bridges $(\bridge_{C_1, \dots, C_s}, \bridge_{D_1, \ldots, D_t})$ (see Definition \ref{def:bridge_comp}) for all such circuits.

We recall the following definition due to Micciancio.

\begin{definition} \label{fche}
    An encryption scheme $\mathcal{E}$ is called \emph{fully composable encryption scheme} \footnote{This is the name coined by Micciancio, based on the consequence of this property described bellow. We chose to use the name complete for its generalisation to bridges.} (FcHE) if for any circuit $C:\mathcal{M}^r\to\mathcal{M}$, the following relation
    
    $$\Dec(sk,\Eval(evk, C, c_1,\ldots, c_r))=C(\Dec(sk,c_1),\ldots, \Dec(sk,c_r))
    $$
    holds for all $c_1,\ldots, c_r\in \mathcal{C}$.
\end{definition}

\begin{remark}\label{rem:fc_from_complete}
It is an immediate consequence of the definition that a scheme is fully composable if and only if every bridge $\bridge_C$ from the familly defined above is a complete bridge (see Definition \ref{complete:bridge}).
\end{remark}

We point out the following immediate consequence, which is consistent with the terminology chosen by Micciancio.

\begin{proposition}
    Let $\mathcal{E}$ be an FcHE scheme and let $C_1,C_2$ be two (multivalued) circuits, $C_1:\mathcal{M}^r \to \mathcal{M}^s$ and $C_2:\mathcal{M}^s \to \mathcal{M}^t$. Then
$$\Dec(sk,\Eval(evk, C_2, \Eval(evk, C_1, c_1, \ldots, c_r)) = C_2(C_1(m_1, \ldots, m_r))$$
for all $c_1, \ldots, c_r \in \mathcal{C}$, where $m_i = \Dec(sk,c_i)$ for all $i = \overline{1,r}$. 
\end{proposition}

In particular, in FcHE schemes, one can homomorphically  compose \footnote{By homomorphic composition we understand sequential evaluation of two (multivalued) circuits on encrypted data.} any circuits. However, the converse is not true:

\begin{example}
    Consider any FHE scheme $\mathcal{E}$ which is constructed using a bootstrapping procedure and such that for every ciphertext $c$, we have $\mathrm{Bootstrap}(c) \in \Enc(pk, \Dec(sk,c))$. In such a scheme, one can homomorphically perform bootstrapping after every operation. Then modify the evaluation algorithm of $\mathcal {E}$ such that for every circuit $C$,
    $$\Eval'(evk, C,  c_1, \ldots, c_r) := \mathrm{Bootstrap}(\Eval(evk, C, c_1, \ldots, c_r)).$$
In such a scheme, one can homomorphically compose every two circuits, as the output of every evaluation algorithm will be a fresh encryption. However, such a scheme is not necessarily a FcHE scheme according to the Definition \ref{fche}.
\end{example}

We obtain Micciancio's result as a corollary to Theorem \ref{thm:Gentry_complete}, as we now explain.
Indeed, recall from the construction of Gentry-type bridges that one can realise the set of secret keys of $\mathcal H$ and its ciphertext space as subsets of $\mathcal M^e$ and respectively $\mathcal{M}^n$, where $\mathcal{M}$ is the plaintext of $\mathcal{E}$. In the same section, we constructed the map $\widetilde{ \Dec}: \mathcal{M}^{e} \times \mathcal{M}^n \to \mathcal M$. For any ciphertext $c$ of $\cH$, we can restrict the second argument of $\widetilde{\Dec}$ to obtain $\widetilde{\Dec}(\cdot, c): \mathcal M^e \to \mathcal M$. 

The encryption scheme $\mathcal H'$ proposed by Micciancio is constructed as follows. The encryption and decryption procedures are the ones from $\mathcal H$. The key generation algorithm uses $\mathrm{KeyGen}$ to obtain a triple  $(sk,pk, evk)$. The algorithm now encrypts each component of $sk=(sk[1],...,sk[e])$, viewed as an element of $\mathcal M^e$, to obtain $(\widetilde{sk[1]},..., \widetilde{sk[e]})\in \mathcal{C}^e$. Finally, it outputs the triple $(sk, pk,evk')$, where $evk'=(evk, \widetilde{sk[1]},..., \widetilde{sk[e]})$. 

For any $C : \mathcal{M}^l \to \mathcal M$, the evaluation algorithm $\mathrm{Eval}'$ works as follows. Let $c_1, c_2, \ldots, c_l \in \mathcal C \subseteq \mathcal M^n$  be any $l$ ciphertexts. Let $g_{(C, c_1, ..., c_l)}: \mathcal M^e \to \mathcal M$ be  defined as $$g_{(C, c_1, ..., c_l)}(sk) = C(  \widetilde{\Dec}(sk,c_1),  \widetilde{\Dec}(sk,c_2), \ldots,  \widetilde{\Dec}(sk,c_l) ).$$ 
The evaluation algorithm of $\mathcal H'$ is given by:

\begin{equation*}
    \Eval'(evk', C, c_1, ..., c_l) := \Eval(evk, g_{(C, c_1, ..., c_l)}, \widetilde{sk[1]},..., \widetilde{sk[e]}).
\end{equation*}

We now use the language of bridges to show that $\cH'$ is a secure FcHE, under the circular security assumption for $\cH$.

The circuit $C$ gives rise as above to a bridge
$$\bridge_{C} : \mathcal H'\,^{(l)} \to \mathcal{H}'.$$
Since $\cH$ and $\cH'$ are identical once we forget the evaluation algorithms, we can view this bridge as a bridge ${\bf Br}_C:\cH^{(l)}\to \cH$. However, this bridge is \emph{not} the bridge associated to the circuit $C$ as above. Actually, it is easy to see that this bridge is the Gentry-type bridge constructed in Section \ref{sec:Gentry} where $\mathcal{E}=\cH^{(l)}$, and $\iota:\cM^l\to\cM$ is defined by $C$, with the difference that the $\KeyGen$ outputs the same secret key for both $\cE$ and $\cH.$ However, the proof of Theorem \ref{thm:Gentry_complete} transports identically, thus the bridge $\bridge_C$ is complete. Using Remark \ref{rem:fc_from_complete}, the scheme $\cH'$ is an FcHE scheme.

As we mentioned in Remark \ref{secure:Gentry}, the authors proved in \cite[Theorem 3]{BLPT23} that any Gentry-type bridge between two secure encryption schemes is also secure. However, this result does not directly apply to $\bridge_C$ because this bridge has a slightly different KeyGeneration algorithm. Namely, here the secret keys of the encryption schemes involved in the bridge are identical, whereas in Gentry-type bridges they are generated independently. In the present situation, the security analysis of the scheme $\cH'$ is much simpler, because it is equivalent to the security analysis of $\cH[evk']$. On the other hand, the security of the latter scheme reduces to the circular security assumption.

%
%
%

\begin{thebibliography}{8}

\bibitem{AFGH06} Ateniese, G.,  Fu, K.,  Green, M., Hohenberger, S.: Improved proxy re-encryption schemes with applications to secure distributed storage, In:  ACM Trans. Inf. Syst. Secur., vol 9, no 1, 1–30 (2006)


 \bibitem{BLPT23} Barcau, M., Lupa\c scu, C., Pa\c sol, V., \c Turca\c s G. C.: Bridges connecting Encryption Schemes, In Bella, G., Doinea, M., Janicke, H. (eds) Innovative Security Solutions for Information Technology and Communications. SecITC 2022. LNCS, vol 13809, pp. 37--64, Springer, Cham (2023).
 
 \bibitem{BR05} Bellare, M., Rogaway, P., \textit{Introduction to Modern Cryptography}, University of California, Notes, \url{https://web.cs.ucdavis.edu/~rogaway/classes/227/spring05/book/main.pdf}.


\bibitem{BGV12} Brakerski, Z., Gentry, C., Vaikuntanathan, V.: (Leveled) fully homomorphic encryption without bootstrapping, ACM Transactions on Computation Theory \textbf{6}(3), No. 13, pp. 1--36 (2014). 

\bibitem{BV11} Brakerski, Z., Vaikuntanathan, V.: Efficient fully homomorphic encryption
from (standard) LWE, \url{http://eprint.iacr.org/2011/34}.

\bibitem{Do03} Dodis, Y., Ivan, A.: Proxy cryptography revisited, In: Proceedings of the Tenth Network
and Distributed System Security Symposium, February 2003.

\bibitem{PWANB16} Phong, L. T., Wang, L., Aono, Y., Nguyen, M. H., Boyen, X.:Proxy Re-Encryption Schemes with Key Privacy from LWE, \url{https://eprint.iacr.org/2016/327.pdf}

\bibitem{Can02} Canetti, R., Krawczyk, H.: Universally composable notions of key exchange and secure channels., In: Knudsen, L.R. (ed.) EUROCRYPT 2002. LNCS, vol. 2332, pp. 337–351. Springer, Heidelberg (2002).


\bibitem{Can07}Canetti, R., Dodis, Y., Pass, R., Walfish, S.: Universally Composable Security with Global Setup, In: Vadhan, S.P. (eds) Theory of Cryptography. TCC 2007. Lecture Notes in Computer Science, vol 4392. Springer, Berlin, Heidelberg. https://doi.org/10.1007/978-3-540-70936-7-4

\bibitem{Can20} Canetti, R.: Universally Composable Security, Journal of the ACM, \textbf{67}{5}, 2020, pag 1--94 


\bibitem{Rasta18} Dobraunig, C., Eichlseder, M., Grassi, L., Lallemand, V., Leander, G., List, E., Mendel, F., Rechberger, C.: Rasta: A Cipher with Low ANDdepth and Few ANDs per Bit, In: Advances in Cryptology, CRYPTO 2018, LNCS, vol. 10991, pp. 662 -- 692. Springer, Cham (2018).

\bibitem{Pasta21} Dobraunig, C., Grassi, L., Helminger, L., Rechberger, C., Schofnegger, M. and Walch, R.: Pasta: A Case for Hybrid Homomorphic Encryption. In Cryptology ePrint Archive (2021).

\bibitem{DN21} Dottling, N., Nishimaki, R.: Universal Proxy Re-Encryption, Cryptology ePrint Archive, Report 2018/840, to appear in PKC '21.

\bibitem{DM15} Ducas, L., Micciancio, D.: FHEW: Bootstrapping Homomorphic Encryption in Less Than a Second, Advances in Cryptology - EUROCRYPT 2015,  Lecture Notes in Computer Science vol. 9056, pp. 617 -- 640. 

\bibitem{Ge10} Gentry, C.: Computing arbitrary functions of encrypted data, Communications of the ACM, \textbf{53}(3), pp.  97 -- 105 (2010).

\bibitem{Go90} Goldreich, O.: A note on computational indistinguishability, Information Processing Letters, \textbf{34}(6), pp. 277 -- 281.

\bibitem{GM82} Goldwasser, S., Micali, S.: Probabilistic encryption and how to play mental poker keeping secret all partial information, In: STOC '82: Proceedings of the fourteenth annual ACM symposium on Theory of computing, pp. 365 -- 377. Association for Computing Machinery, New York, NY (1982).


\bibitem{Lee22} Lee, Y. et al. (2023): Efficient FHEW Bootstrapping with Small Evaluation Keys, and Applications to Threshold Homomorphic Encryption, In: Hazay, C., Stam, M. (eds) Advances in Cryptology, EUROCRYPT 2023. LNCS, vol. 14006, Springer, Cham (2023).


\bibitem{Mic22} Micciancio, D., [FHE\_org], (9th of June, 2022), Fully Homomorphic Encryption: Definitional issues and open problems [Video], Youtube, \url{https://www.youtube.com/watch?v=b24WJyS0dmg} 

\bibitem{MP21} Micciancio, D., Polyakov, Y.: Bootstrapping in FHEW-like Cryptosystems, In: Proceedings of the 9th on Workshop on Encrypted Computing \& Applied Homomorphic Cryptography (WAHC '21). Association for Computing Machinery, New York, NY, USA, 17–-28.


\bibitem{SYY} Sander, T., Young, A., Yung, M.: Non-Interactive CryptoComputing For $NC^1$. In: FOCS '99: Proceedings of the 40th Annual Symposium on Foundations of Computer Science, pp. 554 -- 566, IEEE Computer Society, NW Washington, DC, United States (1999).

\end{thebibliography}
%

\end{document}